\documentclass[11pt, a4paper]{article}
\usepackage
[
        a4paper,
        left=3cm,
        right=3cm,
        top=3cm,
        bottom=3cm,
]
{geometry}

\usepackage{amsmath}
\usepackage{amstext}
\usepackage{amssymb}
\usepackage{amsthm}
\usepackage{pdfpages}	

\usepackage{caption}
\newcommand{\footremember}[2]{
	\footnote{#2}
	\newcounter{#1}
	\setcounter{#1}{\value{footnote}}
}

\newcommand{\footrecall}[1]{
	\footnotemark[\value{#1}]
}

\providecommand{\keywords}[1]{\noindent \textbf{\textit{Keywords: }}#1}

\newcommand{\SSS}{{S^1}}

\newcommand{\R}{\mathbb{R}}
\newcommand{\Z}{\mathbb{Z}}
\newcommand{\N}{\mathbb{N}}

\newcommand{\Prb}{\mathbb{P}}
\newcommand{\EE}[1]{\mathbb{E}\left[#1\right]}
\newcommand{\OT}{{OT}}
\newcommand{\COT}{{C\!\!\;OT}}
\newcommand{\XC}{\mathcal{X}}
\newcommand{\DD}{{\mathcal{D}([0,1))}}
\newcommand{\CC}{{\mathcal{C}(\SSS)}}
\renewcommand{\CC}{{\mathcal{C}_0(\Interval)}}
\renewcommand{\DH}[1]{{D^H_{#1}}}
\newcommand{\LevMed}{\textup{LevMed}}
\newcommand{\Med}{{\textup{Med}}}

\newcommand{\sign}{{\textup{sign}}}
\newcommand{\cov}{{\textup{Cov}}}
\renewcommand{\Bbb}{\mathbb{B}}

\newcommand{\konvW}{\xrightarrow{\;\;\mathcal{D}\;\;}}
\newcommand{\Interval}{{[0,1)}}
\newcommand{\Var}{\textup{Var}}

\newcommand{\BL}[1]{\textup{BL}_1(#1)}

\newcommand{\norm}[1]{\left\|#1\right\|}

\newcommand{ \coloneqq}{{\,:=\,}}

\theoremstyle{plain}
\newtheorem{theorem}{Theorem}[section]

\theoremstyle{definition}

\newtheorem{remark}{Remark}
\newtheorem*{Test}{Circular optimal transport test\,}
\newtheorem{example}{Example}

\providecommand{\keywords}[1]{\textbf{\textit{Keywords: }} #1}

\date{\today}

\title{The Statistics of Circular Optimal Transport}

\author{Shayan Hundrieser \footremember{ims}{\scriptsize Institute for Mathematical
		Stochastics, University of G\"ottingen,
		Goldschmidtstra{\ss}e 7, 37077 G\"ottingen} 
	\and 
	Marcel Klatt \footrecall{ims}{}
	\and 
	Axel Munk \footrecall{ims} \footnote{\scriptsize Max Planck Institute for Biophysical
		Chemistry, Am Fa{\ss}berg 11, 37077 G\"ottingen}}

\renewenvironment{abstract}
 {\small
  \begin{center}
  \bfseries \abstractname\vspace{-.5em}\vspace{0pt}
  \end{center}
  \list{}{
    \setlength{\leftmargin}{.5cm}
    \setlength{\rightmargin}{\leftmargin}
  }
  \item\relax}
 {\endlist}
	
\begin{document}

\maketitle
\begin{abstract}
Empirical optimal transport (OT) plans and distances provide effective tools to compare and statistically match probability measures defined on a given ground space.
Fundamental to this are distributional limit laws and we derive a central limit theorem for the empirical OT distance of circular data.
Our limit results require only mild assumptions in general and include prominent examples such as the von Mises or wrapped Cauchy family. Most notably, no assumptions are required when data are sampled from the probability measure to be compared with, which is in strict contrast to the real line. A bootstrap principle follows immediately as our proof relies on Hadamard differentiability of the OT functional. This paves the way for a variety of statistical inference tasks and is exemplified for asymptotic OT based goodness of fit testing for circular distributions. We discuss numerical implementation, consistency and investigate its statistical power.
For testing uniformity, it turns out that this approach performs particularly well for unimodal alternatives and is almost as powerful as Rayleigh's test, the most powerful invariant test for von Mises alternatives. For regimes with many modes the circular OT test is less powerful which is explained by the shape of the corresponding transport plan.\\ 
\end{abstract}


\section{Introduction}

 Originally formulated by Monge \cite{monge} and later restated and generalized by Kantorovich \cite{kant1942_original} among others, the mathematical theory of optimal transport (OT) nowadays provides a fertile ground for modern research with comprehensive monographs  \cite{rachev1998massTheory, rachev1998massApplications, santambrogio2015optimal, vil03, villani2008optimal}.
OT plans and their associated distances compare probability measures while incorporating the geometry of the underlying ground space. This aspect, often neglected by typical discrepancy measures such as total variation or Kullback-Leibler divergence, has recently put OT in the spotlight of being a highly informative and effective tool for statistical data analysis and  inferential purposes \cite{del1999tests,del2005asymptotics,evans2012phylogenetic,klatt2018empirical, panaretos2019statistical, sommerfeld2018, tameling18,   zemel2019}.
 OT based data analysis for complex and high-dimensional structures
   has further been encouraged by recent computational progress \cite{cuturi13, cuturi18} paving the way for a variety of applications as diverse as genetics \cite{evans2012phylogenetic},  computational biology \cite{klatt2018empirical,Schiebinger19, Weitkamp2020},   signal processing \cite{kolouri2017optimal}, image retrieval \cite{rabin2008_CEMD_CompLocalFeatures, rubner2000}, fingerprint identification \cite{sommerfeld2018} and procrustes analysis \cite{zemel2019}, among others.
 
A key benefit of OT is its intuitive interpretation as the minimum effort of transporting mass from one distribution to another. More precisely, given two probability measures $\mu, \nu$ on a ground space $\XC$ and a cost function $c\colon \XC\times \XC \rightarrow [0,\infty)$ the OT distance between $\mu$ and $\nu$ is defined as 
\begin{equation}\label{eq:OT_generalFormulation}
  \OT(\mu, \nu) \coloneqq \inf_{\pi}\int_{\XC\times \XC}c(x,y) d\pi(x,y). 
\end{equation}
The infimum is taken over all probability measures on the product space $\XC\times \XC$ whose marginals coincide with $\mu$ and $\nu$.

In many applications the population measure $\mu$ is often not available but instead access to a finite set of independent and identically distributed (i.i.d.) random variables $X_1, \dots, X_n \sim\mu$ is given. Hence, $\mu$ is estimated by the empirical probability measure \begin{equation}
  \hat \mu_n \coloneqq \frac{1}{n}\sum_{i = 1}^{n}\delta_{X_i},\label{eq:empMeasure}
\end{equation} 
and yields the empirical plug-in estimator $\OT(\hat \mu_n, \nu)$ for the unknown population distance  $\OT(\mu, \nu)$. At this point and for simplicity, we assume $\nu$ to be known. The generalization to the two sample scenario where additionally $\nu$ is estimated is analogous (see Remark \ref{rem:TwoSamples}).
Although for computation the OT problem can be cast as a linear program, for many real-world applications the efficient computation of OT distances still is a delicate issue and the development of improved algorithmic solutions is a highly active field of research  \cite{altschuler2017near,dvurechensky2018computational,cuturi18,schmitzer_sparse_2016,Schrieber2017}. By all means, an exceptional case is given on the real line where for certain cost functions explicit solutions for OT distances exist. For instance, 
 for Euclidean costs $c(x,y) = |x-y|$ it is well-known 
 that the OT distance between two probability measures $\mu$ and $\nu$ on $(\R, \mathcal{B}(\R))$ is given by \begin{equation}
   \OT_{\R}(\mu, \nu) = \int_{-\infty}^\infty|F_{\mu}(t) - F_{\nu}(t)|dt,
\label{eq:OT_real_formula}
\end{equation}
where $F_\mu, F_\nu$ denote the respective cumulative distribution functions. Similar formulas in terms of  quantile functions $F^{-1}_\mu, F^{-1}_\nu$ exist for costs which are given by a convex function of the Euclidean distance \cite{vil03, villani2008optimal}. 
This eases the computation but also the statistical analysis of $\OT(\hat \mu_n, \nu)$ for the real line case substantially. Applications include 
 goodness of fit testing 
and other tools for inferential purposes \cite{del1999tests,del1999central, del2005asymptotics,delbarrio2019,Munk98}. The underlying distributional limit theory can become rather involved as the extreme quantiles of $F_\mu$ have to be controlled   (\eqref{eq:OT_real_formula} is a notable exception) \cite{berthet2020exact,bobkov2019one, del2000contributions,del1999tests,del1999central,Munk98}.
\begin{figure}[!t]
\begin{center}
\includegraphics[width=0.9\textwidth, trim=0 18 -18 12, clip]{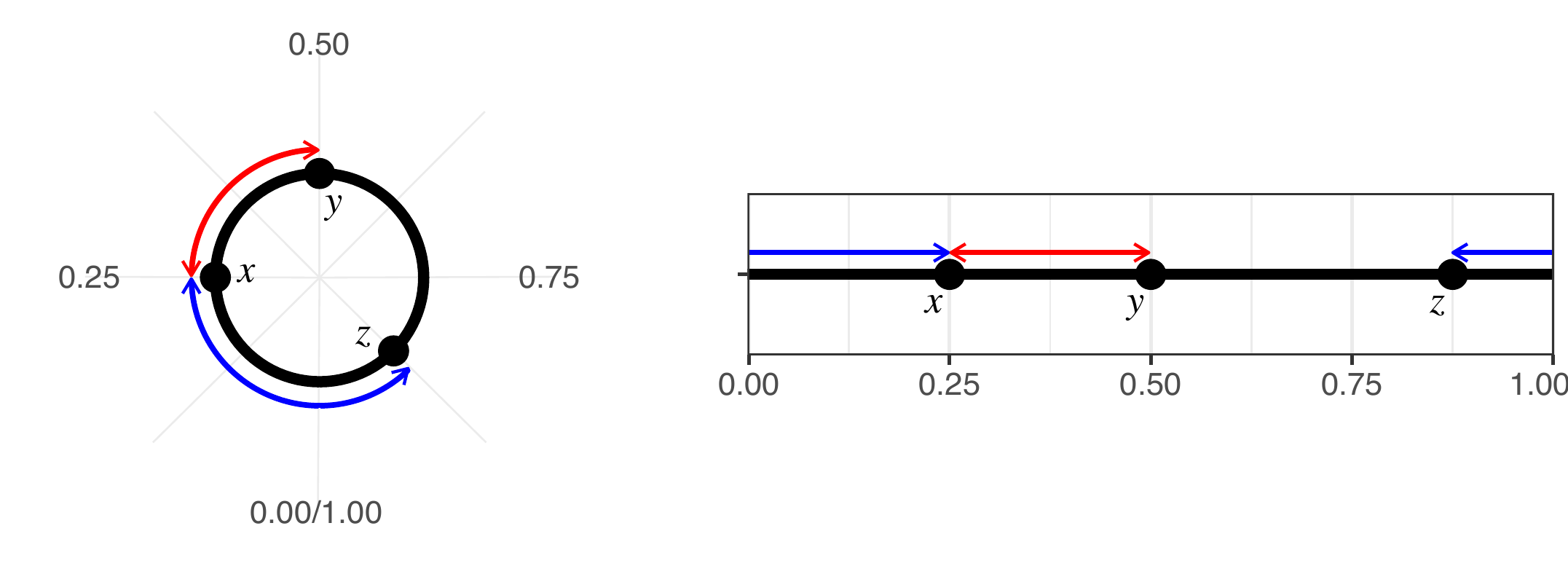}
\caption{\textbf{Geodesics on $\SSS$ and intrinsic metric $\rho_\SSS$.} 
\textbf{Left:} Shortest geodesics on $\SSS$ connecting the points $x$ and $y$ (red) as well as $x$ and $z$ (blue). 
\textbf{Right:} Same geodesics on the interval $\Interval$. The length of the geodesic connecting two points coincides with the distance of these points with respect to $\rho_\SSS$.}
\label{fig:CircularDistance}
\end{center}
\end{figure}
In this work, we investigate statistical properties of circular OT (COT) and
  derive limit distributions of the empirical COT distance extending the theory on limit laws for the real line to circular data.
This complements a considerable amount of contemporary research concerned with the analysis of circular data relevant to applications in biology \cite{batschelet1981circular,landler2018circular}, meteorology and climate research \cite{Hundrieser2020Smeariness}, environmental science \cite{kim2013three,sengupta2006} and image retrieval \cite{rabin2008_CEMD_CompLocalFeatures}, to mention a few. For a comprehensive treatment we refer to \cite{fisher1995statistical, jammalamadaka2001topics, mardia2000directional}. More recent advances  on directional statistics are summarized in \cite{Garcia2018OverviewTestsHyptersphere, Pewsey2021}.   Our work is motivated from the observation that the COT distance appears in a particular intuitive closed form when comparing and analyzing circular distributions.

In the following, we parametrize the circle $\SSS=\R/\Z$ by the set $\Interval$ equipped with the intrinsic metric also known as the geodesic distance (see Figure \ref{fig:CircularDistance}) \begin{equation}\label{eq:IntrinsicMetricCircle}
   \rho_\SSS(x,y) \coloneqq \min(|x-y|, 1 - |x-y|) \quad \forall x,y\in \SSS.
\end{equation}
\;\!\!\!Moreover, we denote by $\mu, \nu$ two Borel probability measures on $\SSS$ with respective cumulative distribution functions $F_\mu, F_\nu: \Interval \rightarrow [0,1]$ defined as \begin{equation}\label{eq:CDF}
  F_\mu(t) \coloneqq \mu([0,t]), \quad F_\nu(t) \coloneqq \nu([0,t]) \quad \forall t \in \Interval.
\end{equation}
\begin{figure}[!b]
\begin{center}
\includegraphics[width =0.9\textwidth, trim = 20 0 10 20, clip]{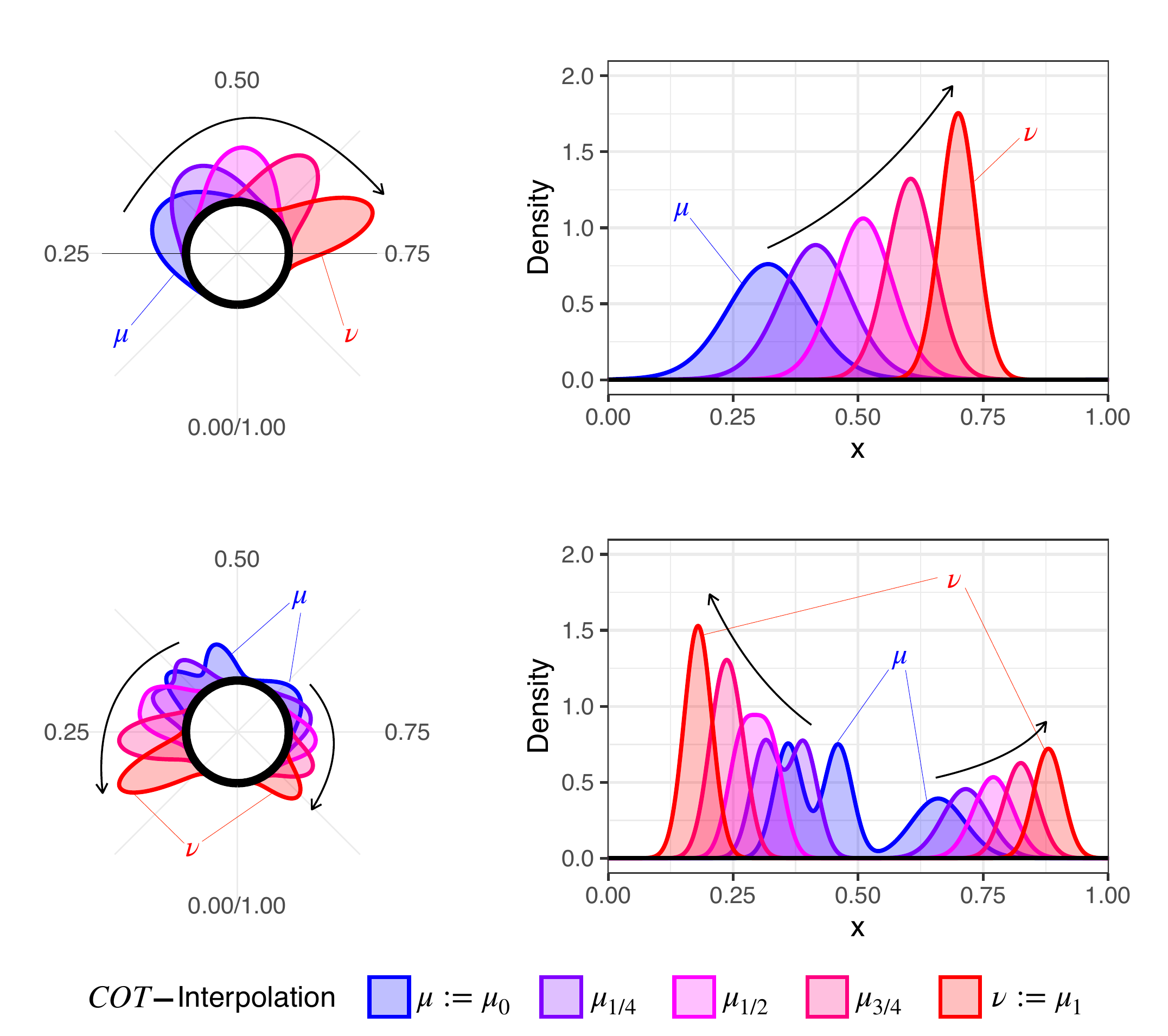}
\caption{\textbf{Circular optimal transport on $\SSS$.} Mass transportation from $\mu$ (blue) to $\nu$ (red) represented via displacement interpolation $\mu_t$ for $t \in \{0, 1/4, 1/2, 3/4,1 \}$ with respect to $\COT$ (for details see \cite{santambrogio2015optimal}) on the circle (left) and its corresponding cartesian plot (right). \textbf{Top:} Probability measures $\mu$, $\nu$ with unimodal characteristics. \textbf{Bottom:} Probability measures $\mu$, $\nu$ with multimodal characteristics.}
\label{fig:OT_Interpolation}
\end{center}
\end{figure}
\;\!\!\!Our analysis for the COT distance $\COT(\mu, \nu)$, defined as in \eqref{eq:OT_generalFormulation} with cost functional $c(x,y) = \rho_\SSS(x,y)$, relies on the explicit formula 
\begin{equation}
	\COT(\mu, \nu) = \inf_{\alpha \in \R} \int_{0}^1\big|F_\mu(t)- F_\nu(t) - \alpha\big|dt,
\label{eq:OTFormula}
\end{equation}
proven initially for discrete probability measures $\mu, \nu$ on $\SSS$ by Werman et al. \cite{werman1985distance} and later generalized to arbitrary probability measures by Delon et al. \cite{delon2010fast}. 
This formula shows some analogy to the expression for $\OT_\R$ from \eqref{eq:OT_real_formula}. The additional infimum over $\alpha$ arises from the ambiguity of how to register cumulative distribution functions on a circle. In particular, one needs to set a proper origin. Given the optimal choice of the origin for $\mu$ and $\nu$ (i.e. the minimizing element $\alpha$ in \eqref{eq:OTFormula}), the COT problem essentially reduces to the OT problem on the interval $\Interval$. For an illustration we refer to Figure \ref{fig:OT_Interpolation}.

Exploiting the representation \eqref{eq:OTFormula} in conjunction with weak convergence ($\konvW$) of the empirical process $\sqrt{n}\left(F_{\hat \mu_n} - F_\mu\right)$ towards an $F_\mu$-Brownian Bridge $\Bbb_{F_\mu}$ as $n$ tends towards infinity \cite[Theorem 14.3]{billingsley1999convergence}, we prove in Theorem~\ref{them:LimitLawEqualDistr} that
\begin{equation}
  \sqrt{n}\,\COT(\hat\mu_n, \mu) \konvW \inf_{\alpha \in \R}\int_{0}^{1} \left| \Bbb_{F_\mu}(t) - \alpha\right|dt, \quad \text{ as } n \rightarrow \infty. \label{eq:LimitLawEqualDists}
\end{equation}
Note that in \eqref{eq:LimitLawEqualDists} the data is sampled from the same probability measure $\mu = \nu$ it is compared with.
Most notably, in this situation we do not require \emph{any} assumptions on $\mu$ for this result to be valid. Our theory also holds for the two-sample case, where two i.i.d. samples stem from the same probability measure and their empirical counterparts are compared using COT distances.   In contrast, for $\mu \neq \nu$ we require additional assumptions (see (A1), (A2), (A3) in Section \ref{sec:LimitLawsOTD}) to obtain a normal limit
\begin{equation}
  \sqrt{n}\big(\COT_\SSS\left(\hat \mu_n, \nu\right) - \COT_\SSS( \mu, \nu) \big) \konvW \mathcal{N}\left(0,\sigma^2_{\mu|\nu}\right), \quad \text{ as } n \rightarrow \infty,\label{eq:LimitLawDifferentDists}
\end{equation}
where $\mathcal{N}\left(0,\sigma^2_{\mu|\nu}\right)$ denotes a centered Gaussian law with variance $\sigma^2_{\mu|\nu}$ that can be computed explicitly (Theorem \ref{them:LimitLawDifferentDistr}). Furthermore, these results are extended to bootstrap consistency (Theorem \ref{them:Bootstrap_Consistency_EqualDistr}). More precisely, 
for $\mu \neq \nu$ we prove under suitable assumptions that the naive $n$-out-of-$n$ bootstrap is consistent for  \eqref{eq:LimitLawDifferentDists}. In the setting $\mu = \nu$, we find that this bootstrap procedure fails but instead the $m$-out-of-$n$ bootstrap with $m = o(n)$ is consistent for~\eqref{eq:LimitLawEqualDists}. We emphasize that much of our asymptotic theory turns out to be simpler than the usual case on the real line as the circle is a compact manifold.

Based on these asymptotic statements, we propose the COT test (COTT)
(see Section \ref{sec:TestingGoF}) investigating 
the hypothesis that a given sample stems from a particular probability measure $\mu_0$ on $\SSS$. We employ our goodness of fit approach 
  to test for uniformity and compare  it to prominent tests by Rayleigh \cite{rayleigh1880}, Kuiper \cite{Kuiper1960TestsCR}, Watson \cite{watson1961goodness}, Rao \cite{rao1969some} as well as more recently proposed test methods by Pycke \cite{pycke2010some}. It turns out that the COTT for uniformity performs particularly well for unimodal alternatives. 
For multimodal alternatives, the COTT is less powerful, even though it outperforms other well-known tests specifically designed for unimodal alternatives. 
In short, if it is expected that the alternative distribution has only a few modes, we recommend the COTT for goodness of fit testing.

The outline of this paper is as follows.  In Section \ref{sec:CircOT}, we assess the explicit formula \eqref{eq:OTFormula} for COT, verify the existence of a minimizer $\alpha$, give an alternative characterization, and provide a computational scheme through discretization which relies on an alternative representation of the optimal choice for $\alpha$ that is based on ordering sets. More precisely, we observe that the exact quantity $COT(\mu, \nu)$ can be approximated up to an error of $\mathcal{O}(1/D)$ for $D\in \N$ with a computational effort of $\mathcal{O}(D \log(D))$ operations. This is in line with findings by \cite{delon2010fast}. Our main contribution is given in Section \ref{sec:LimitDistr} and concerned with distributional limits. We start with a short overview of required results from empirical process theory and weak convergence. Our main results on limit laws of empirical OT distances are stated in Section \ref{sec:LimitLawsOTD} and are extended in Section \ref{sec:LimitLawsBootstrapOTD} to bootstrap consistency. In Section \ref{sec:SimulationForLimitLaws}, the finite sample accuracy of our asymptotic results is analyzed in Monte Carlo Simulations. 
 Section \ref{sec:TestingGoF} is dedicated to formalizing COTT and proving its asymptotic consistency. We then examine the statistical power of COTT, compare it to other prominent tests, and give an intuitive explanation for its performance based on the nature of optimal transport. Finally, we summarize our results in Section \ref{sec:summary} and discuss open questions for future research.
 
We provide an \textsc{R}-package \cite{Rcore2020} \texttt{circularOT} for circular data analysis with OT. Besides computation of the COT distance between data samples the package includes an implementation of the COTT for uniformity as well as a bivariate bootstrap based COTT to assess whether two samples stem from the same distribution. The package is available at \emph{https://gitlab.gwdg.de/shundri/circularOT}.
Furthermore, an overview of this work with illustrations and animations is available at \emph{https://stochastik.math.uni-goettingen.de/cot}.

\section{Circular Optimal Transport: Alternative Representation and Numerical Computation}\label{sec:CircOT}
The representation \eqref{eq:OTFormula} (see \cite{delon2010fast, rabin2008_CEMD_CompLocalFeatures, werman1985distance}) reveals the COT distance with respect to the metric $\rho_\SSS$ (see \eqref{eq:IntrinsicMetricCircle}) as an  optimization problem  in only one parameter $\alpha$.
Notably, for a given measurable function $f \colon \Interval \rightarrow \R$ it follows that the mapping $\alpha \mapsto \int_\Interval |f(t) - \alpha|dt$ is convex and coercive\footnote{A function $g\colon \R \rightarrow \R$ is called coercive if $g(x)\rightarrow \infty$ as $|x| \rightarrow \infty$.}. Hence, there exists a compact set of global minimizers for $\alpha$ among which we consider the smallest element and refer to it 
as \emph{level median}
\begin{equation}\label{eq:LevelMedianDefinition}
  \LevMed (f) \coloneqq \min\left\{\arg\min_{\alpha \in \R} \int_{0}^1\big|f(t) - \alpha\big|dt\right\}.
\end{equation}

\begin{figure}[b]
\begin{center}
\includegraphics[width=0.9\textwidth, trim = 0 8 0 9, clip]{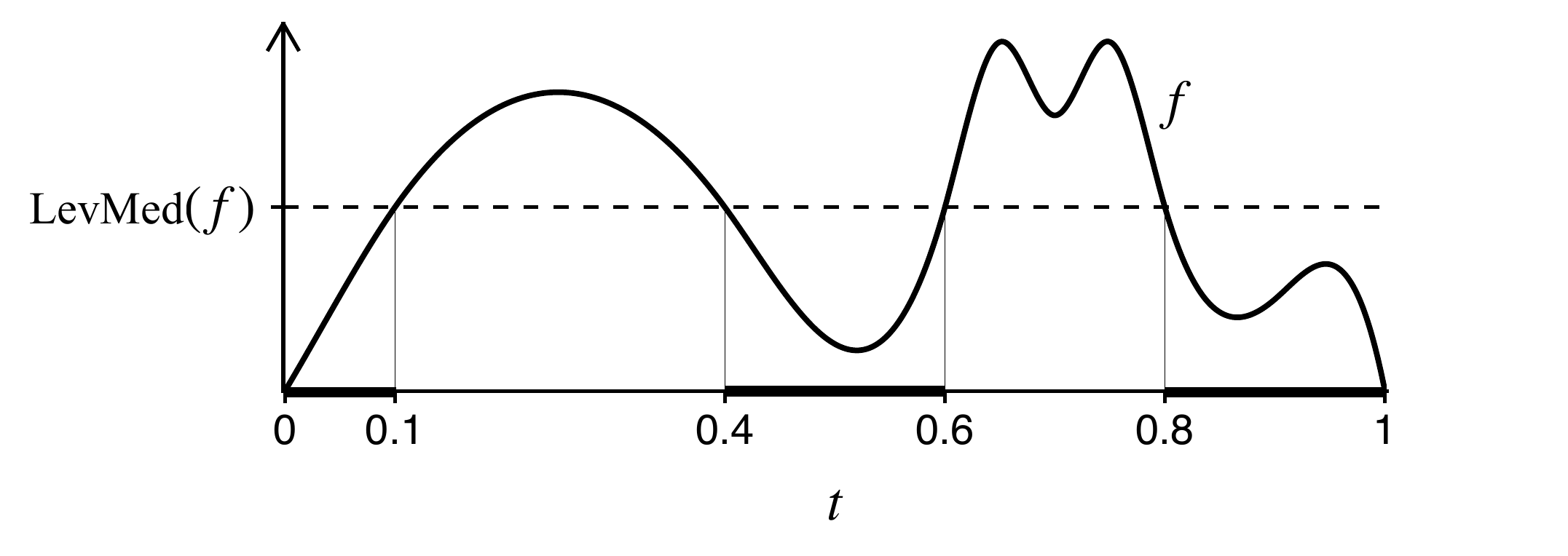}
\caption{\textbf{Level median for a function $f$ on $\Interval$.} Visualization of level median (dashed) for a function $f$ (solid). The set of elements $x$ in $\Interval$ for which $f(x)$ is smaller than the level median is depicted by a thickened x-axis. The Lebesgue measure of this set is equal to~$1/2$.}
\label{fig:LevelMedian}
\end{center}
\end{figure}

\noindent Intuitively, the level median $\LevMed(f)$  describes the median of the image of $f$, i.e. its levels with respect to Lebesgue measure. More precisely, it is shown by Bivens and Klein \cite{Median2015} that $$ \LevMed(f) = \inf\left\{t \in \R \colon \lambda(\{ x \in \Interval \colon f(x) \leq t \}) \geq 1/2\right\},$$
where $\lambda$ denotes the Lebesgue measure (see Figure \ref{fig:LevelMedian}). Let us emphasize that the level median is not to be confused with the classical (statistical) median of a continuous random variable $X$ on $\Interval$ with cumulative distribution function $F_X$ which is defined as $$\Med(X)= \inf\left\{ t \in \R \colon F_X(t) \geq 1/2\right\}.$$
Notably, $\Med(X)$ might attain any value in $[0,1)$. In contrast, for the particular case when $f=F_X$ is a cumulative distribution function, monotonicity and continuity of $F_X$ always yield that $\LevMed(F_X) = 1/2$. Concluding, the COT distance can be expressed by$$\COT(\mu, \nu) = \int_0^1 \left| F_\mu(t) - F_\nu(t) - \LevMed(F_\mu-F_\nu) \right|dt.$$
Intuitively, this formula arises by setting a proper origin for cumulative distributions functions $F_\mu$, $F_\nu$ at $\LevMed(F_\mu - F_\mu)$ and then 
employing formula \eqref{eq:OT_real_formula} for OT distances on $\R$.  

For computation of the COT distance, we define for $D \in \N$ the discretized probability measures $$ \tilde{\mu}_D \coloneqq \sum_{i = 0}^{D-1} \mu\left(\left[\frac{i}{D}, \frac{i+1}{D}\right)\right)\delta_{i/D},  \quad \tilde{\nu}_D \coloneqq \sum_{i = 0}^{D-1} \nu\left(\left[\frac{i}{D}, \frac{i+1}{D}\right)\right)\delta_{i/D}.$$
In particular, it is easy to see using monotone couplings \cite{santambrogio2015optimal} that $\COT(\mu, \tilde{\mu}_D)\leq 1/D$ and $\COT(\nu, \tilde{\nu}_D)\leq 1/D$. Furthermore, since the COT distance defines a metric on the space of probability measures on $\SSS$ \cite{vil03,villani2008optimal}, we obtain by the triangle inequality that $$ \left|\COT(\mu, \nu)  - \COT(\tilde\mu_D, \tilde\nu_D)   \right| \leq \COT(\mu, \tilde{\mu}_D) +  \COT(\nu, \tilde{\nu}_D) \leq \frac{2}{D}.$$
Hence, the quantity $\COT(\tilde\mu_D, \tilde\nu_D)$ approximates the exact COT distance between $\mu$ and $\nu$ up to an error of size $\mathcal{O}(1/D)$. In particular, it holds that 
$$ \begin{aligned} \COT(\tilde\mu_D, \tilde\nu_D) &= \int_0^1 \left| F_{\tilde \mu}(t) - F_{\tilde\nu}(t) - \LevMed(F_{\tilde\mu_D}-F_{\tilde\nu_D}) \right|dt \\
&=\frac{1}{D} \sum_{i = 1}^{D}\left| F_\mu(i/D) - F_\nu(i/D) - \LevMed\big( F_{\tilde\mu_D}- F_{\tilde\nu_D}\big)\right|,\\
\end{aligned}$$
where the level median is characterized by 
$$\LevMed\left(F_{\tilde\mu_{D}}- F_{\tilde\nu_{D}}\right)=\Med\Big(\left\{ F_\mu(i/D)- F_\nu(i/D) \colon i = 1, \dots, D \right\}\Big).$$
Consequently, the COT distance between the discretized measures $\tilde \mu_D$ and $\tilde \nu_D$ can be calculated with a computational effort of $\mathcal{O}(D \log(D))$ arithmetic operations. This rate is in line with algorithms provided by Delon et al. \cite{delon2010fast} for the computation of OT distances in case of more general cost functions. More precisely, for probability measures $\mu, \nu$ supported on $N$ points their method requires $\mathcal{O}(N |\log(\epsilon)|)$ to approximate $\COT(\mu, \nu)$ up to an error of $\epsilon>0$.

\section{Limit Distributions}\label{sec:LimitDistr}

For an i.i.d. sample $X_1, \dots, X_n \sim \mu$ we consider the empirical probability measure $\hat\mu_n$ introduced in \eqref{eq:empMeasure} and define analogously to the cumulative distribution function $F_\mu$ from \eqref{eq:CDF} the associated empirical cumulative distribution function $F_{\hat \mu_n}$. We are interested in the asymptotic fluctuation of the empirical cumulative distribution function $F_{\hat \mu_n}$ around $F_\mu$, for which we follow standard literature \cite{billingsley1999convergence}. Let $\DD$ be the Banach space of right-continuous functions on $\Interval$, for which left limits exist (c\` adl\` ag-functions), i.e.  $$ \DD \coloneqq \left\{ f \colon \Interval \rightarrow \R \colon \text{ $f$ is c\` adl\` ag, \;} \sup_{t \in \Interval}|f(t)| < \infty\right\},$$
equipped with supremum norm $\norm{f}_\infty \coloneqq \sup_{t \in \Interval}|f(t)|$.  By Donsker's Theorem it follows that the empirical process $\sqrt{n}\big(F_{\hat \mu_n}- F_{\mu}\big)$ converges weakly in $\DD$ for $n \rightarrow \infty$ towards an $F_\mu$-Brownian bridge \cite[Theorem 14.3]{billingsley1999convergence} 
\begin{equation}\label{eq:DonskerTheorem}
  \sqrt{n}\big(F_{\hat \mu_n}- F_{\mu}\big) \konvW \Bbb_{F_\mu} \coloneqq \big(\Bbb_{F_\mu(t)}\big)_{t \in \Interval} \quad \text{ in } \DD,
\end{equation}
where $\Bbb_{F_\mu} \coloneqq \big(\Bbb_{F_\mu(t)}\big)_{t \in \Interval}$ is a centered Gaussian process with covariance $$ \cov[\Bbb_{F_\mu(s)}, \Bbb_{F_\mu(t)}] = \min(F_\mu(s),F_\mu(t)) - F_\mu(s)F_\mu(t) \quad \forall s,t \in [0,1).$$
In the following, we employ the asymptotic statement \eqref{eq:DonskerTheorem} in conjunction with the continuous mapping theorem and the functional delta method. In particular, our main statements follow from that.

\subsection{Limit Laws for the Empirical Circular Optimal Transport Distance}
\label{sec:LimitLawsOTD}
For the formulation of our main results, we consider
 the one-sample case, i.e. $\COT(\mu, \nu)$ is approximated by the empirical plug-in estimator $\COT(\hat \mu_n, \nu)$. The two-sample case is analogous, see Remark \ref{rem:TwoSamples}. We start with the setting that $\mu$ is estimated by its empirical counterpart $\hat \mu_n$ in COT distance. 
\begin{theorem}\label{them:LimitLawEqualDistr} Let $\mu$ be a probability measure on $\SSS$ and denote by $\hat \mu_n$ its empirical probability counterpart based on i.i.d. samples  $X_1, \dots, X_n \sim \mu$.
	As the sample size $n$ tends to infinity it holds that $$ \sqrt{n}\,\COT(\hat\mu_n, \mu) \konvW \inf_{\alpha \in \R}\int_{0}^{1} \left| \Bbb_{F_\mu}(t) - \alpha\right|dt=\int_{0}^{1} \left| \Bbb_{F_\mu}(t) - \LevMed(\Bbb_{F_\mu})\right|dt.$$
\end{theorem}

\begin{proof}
According to \cite[Lemma 3]{avenski2019LimitMonotoneArrangement} for any $f \in \DD$ and a positive constant $a>0$ it holds that $\LevMed(af) = a\LevMed(f)$. Further, \cite[Theorem 1]{avenski2019LimitMonotoneArrangement} implies that $\LevMed \colon \DD \rightarrow \R$ is a contraction. Hence, at the constant  zero-function $f_0 \equiv 0$  the 
	level median is directionally Hadamard differentiable  (see  \cite{Roemisch04} for a definition) with non-linear Hadamard  derivative $$\DH{f_0} \LevMed \colon \DD \rightarrow \R, \quad \Delta \mapsto \LevMed(\Delta).$$
		As a consequence, by the functional delta method \cite{Roemisch04} in conjunction with \eqref{eq:DonskerTheorem} it follows for $n \rightarrow \infty$ that $$ \sqrt{n}\Big(F_{\hat \mu_n} - F_{\mu} - \LevMed(F_{\hat \mu_n}  - F_\mu)\Big) \konvW \Bbb_{F_\mu} - \LevMed(\Bbb_{F_\mu})\quad \text{ in } \DD.$$
		An application of the continuous mapping theorem \cite[Theorem 1.3.6]{van1996weak}  for the continuous operator  $\int_0^1|\cdot |dt \colon \DD \rightarrow \R$ yields for $n \rightarrow \infty$ that $$\sqrt{n}\,\COT(\hat\mu_n, \mu)\konvW \int_{0}^{1} \left| \Bbb_{F_\mu}(t) - \LevMed(\Bbb_{F_\mu})\right|dt.$$
		Finally, the assertion on the different representation of the limit law in terms of an infimum follows by definition of the level median.
\end{proof}

To characterize the limit law of the empirical estimator $\COT(\hat\mu_n, \nu)$ around  $\COT(\mu, \nu)$ for $\mu \neq \nu$, more care is required and we need the following assumptions. 
\begin{itemize}
	\item[\textbf{(A1)}] The probability measures $\mu, \nu$ have  a continuous density on $\SSS$.
	\item[\textbf{(A2)}] There are only finitely many positions where the slope of $(F_\mu - F_\nu)$ is zero.
	\item[\textbf{(A3)}] 
	For all intersections $\{t_1, \dots, t_N\}$ between $(F_\mu - F_\nu)$ and $\LevMed(F_\mu - F_\nu)$ \quad \quad it holds that $(F_\mu-F_\nu)'(t_i) \neq 0$ for all $i \in \{1, \dots, N\}$.   
\end{itemize}
	These assumptions ensure that the level median functional is Hadamard differentiable at  $F_\mu- F_\nu$ for perturbations given by continuous functions \cite{chernozhukov2010quantile}.  For an illustration of assumption (A3), we refer to Figure \ref{fig:AssumptionA3}. Notably, for cumulative distribution functions $F_\mu \neq F_\nu$ which can be extended analytically onto the complex plane it follows, by compactness of $\SSS$ and uniqueness theorem for analytic functions \cite{bak2010complex}, that the derivative of $(F_\mu - F_\nu)$ only coincides  with zero only finitely many times on $\Interval$. 	Hence, such analytic setting implies (A2).   In fact, many pairs of distributions on $\SSS$ fulfill all three assumptions. 
		 \begin{figure}[t]
\begin{center}
\includegraphics[width=0.9\textwidth, trim = 0 5 0 5, clip]{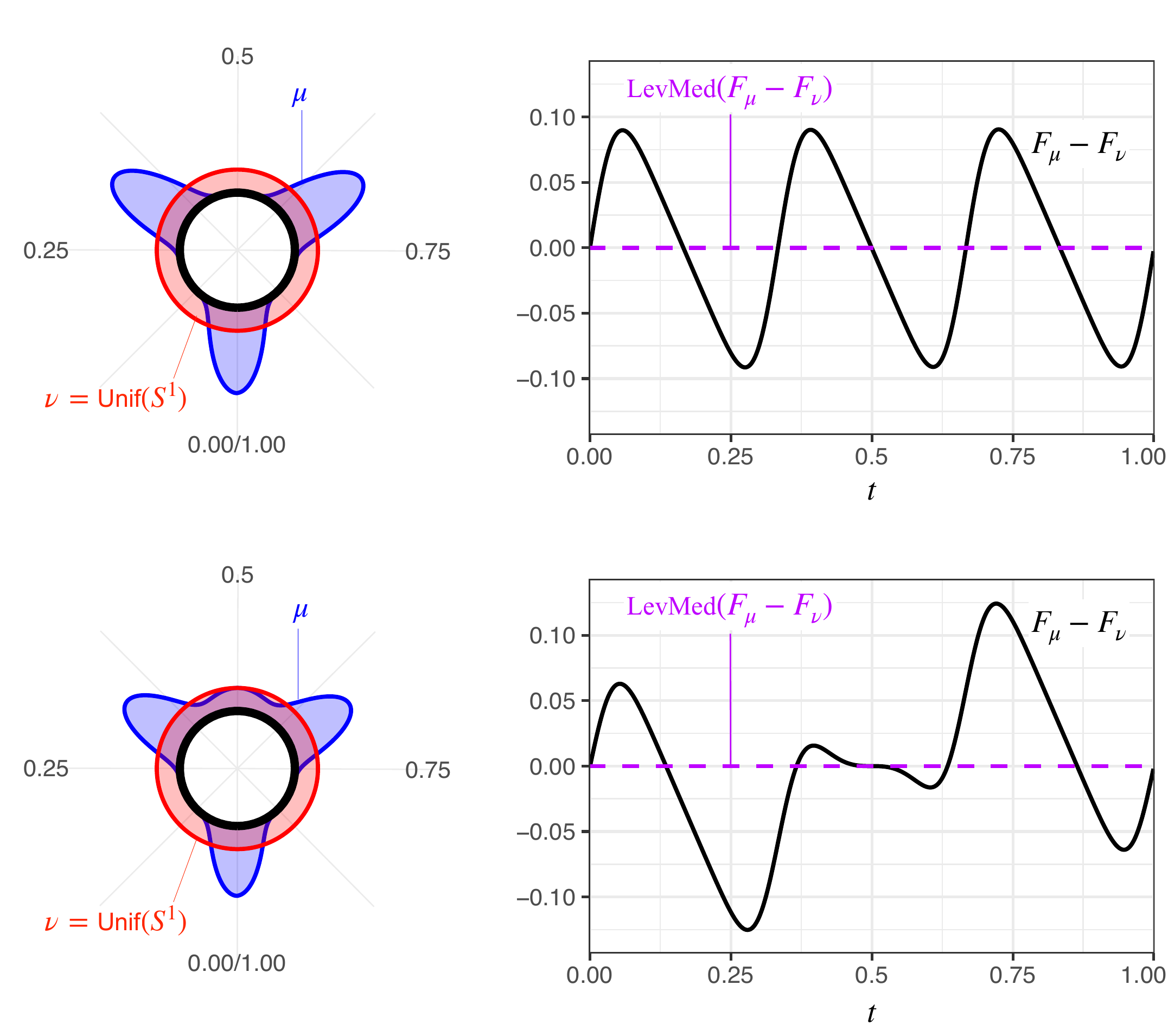}
\caption{\textbf{Example and counterexample for assumption (A3).} \textbf{Top:} Circular density plot for probability measures $\mu, \nu$ on $S^1$ (left). Difference of cumulative distribution functions $F_\mu-F_\nu$ for $\mu, \nu$ (solid, black) in a cartesian plot (right). The associated level median $\LevMed(F_\mu-F_\nu)$ (dashed, purple) is equal to~$0$. Assumption (A3) is valid. \textbf{Bottom:}  The density of $\mu$ coincides at $t = 0.5$ with the density of $\nu$ (left). Assumption (A3) is not satisfied since $(F_\mu-F_\nu)'(t) = 0$ for $t= 0.5$.}
\label{fig:AssumptionA3}
\end{center}
\end{figure}

\begin{example}\label{expl:vonMisesDistributions}
	The parametric family of von Mises distributions \cite{mardia2000directional} is characterized by the class of densities   of the form 
	$$ P_{\theta, \kappa}^{vM}(t) \coloneqq C(\kappa)\exp\big(\kappa \cos(2\pi (t -\theta) \big) \quad \forall x \in \Interval,$$
	 for $\theta \in \Interval, \kappa \in [0, \infty)$ where $C(\kappa)$ denotes the normalization constant. We note  that the density $P_{\theta, \kappa}$ can be extended analytically onto the complex plane for any choice of parameters. Hence, for von Mises distributions $\mu,\nu$ with parameters $(\theta,\kappa) \neq (\tilde\theta, \tilde\kappa)$ assumption (A2) also holds. 
	This yields that $(F_\mu - F_\nu)$ is nowhere constant on $\Interval$ which implies the strict inequalities 
	$$ \min_{t \in \Interval}\big(F_\mu(t) - F_\nu(t)\big)<\LevMed(F_\mu - F_\nu)< \max_{t \in \Interval}\big(F_\mu(t) - F_\nu(t)\big).$$
	 To verify (A3), we prove that $(F_{\mu} - F_{\nu})'(t) = 0$ is satisfied only at the maximum and the minimum of $(F_{\mu} - F_{\nu})$.
	 For this purpose, we note that the equation
	 $$ C(\kappa)\exp\big(\kappa \cos(2\pi (t -\theta) \big) = C(\tilde\kappa)\exp\big(\tilde\kappa \cos(2\pi (t - \tilde \theta) ) \big)$$
	can be equivalently written for some constants $A,C\in \R\,, B \in \Interval$ depending on  $\theta, \tilde\theta, \kappa, \tilde\kappa$ as
\begin{equation}
  \begin{aligned}
			0  &= A \cos(2\pi ( t - B)) + C.
	\end{aligned} \label{eq:VonMisesAssumptionA3}
\end{equation}
		Since at least two solutions exist for \eqref{eq:VonMisesAssumptionA3}, we obtain that $A\neq 0$ which shows that these two solutions are the only ones and verifies the validity of assumption (A3). 
\end{example}

\begin{remark}
With analogous arguments the assumptions (A1), (A2), and (A3) can also be verified for cardioid distributions or wrapped Cauchy distributions \cite{mardia2000directional}.	
\end{remark}
The main result for estimation of $\COT(\mu, \nu)$ by $\COT(\hat \mu, \hat \nu_n)$ for $\mu \neq \nu$ now reads as follows. 
		
	\begin{theorem}\label{them:LimitLawDifferentDistr}
		Let $\mu \neq \nu$ be two probability measures on $\SSS$ and suppose that assumptions (A1), (A2), and (A3) are fulfilled. Denote by $\hat \mu_n$ the empirical probability measure based on i.i.d. samples $X_1, \dots X_n \sim \mu$. 
		As the sample size $n$ tends to infinity it holds that
		\begin{equation*}
			\sqrt{n}\,\Big(\COT(\hat\mu_n, \nu) - \COT(\mu, \nu)\Big)\konvW \mathcal{N}\big(0,\sigma_{\mu|\nu}^2\big),
		\end{equation*}
		where $\mathcal{N}\big(0,\sigma_{\mu|\nu}^2\big)$ denotes a centered Gaussian distribution with variance $\sigma_{\mu|\nu}^2$.
Further, let $\{t_1, \dots, t_N\}$ be the intersections between $F_\mu-F_\nu$ and $\LevMed(F_\mu-F_\nu)$, set $\,t_0 \coloneqq 0$, $t_{N+1}\coloneqq 1$, and define $H_{\mu, \nu}(t) \coloneqq \sign\Big(F_\mu(t) - F_\nu(t) - \LevMed(F_\mu - F_\nu)\Big)$ for all $t \in \Interval$. Then the variance $\sigma^2_{\mu|\nu}$ is characterized by $$\begin{aligned}\sigma^2_{\mu|\nu} &\coloneqq \Var\left[\int_{0}^{1}H_{\mu, \nu}(t)\Bbb_{F_\mu}(t) 
			 dt \right] 
			 = \sum_{i = 0}^N \int_{t_i}^{t_{i+1}} \int_{t_i}^{t_{i+1}} F_\mu(s\wedge \tilde s) ds d\tilde s \\
			 &+ 2 \sum_{i = 1}^{N}\sum_{j=0}^{i-1}H_{\mu, \nu}\bigg(\frac{t_{i+1} + t_i}{2}\bigg)H_{\mu, \nu}\left(\frac{t_{j+1} + t_j}{2}\right)(t_{i+1} - t_{i})\int_{t_j}^{t_{j+1}}F_\mu(s)ds\\
			 &- \left( \int_0^1 H_{\mu, \nu}(s)F_\mu(s)ds \right)^2.
	\end{aligned}$$
	\end{theorem}
	
	\begin{proof}
		Denote by $\CC\subset \DD$ the subspace of continuous functions $f\colon\Interval \rightarrow \R $  such that $f(0) = 0$ and $\lim_{x \nearrow 1}f(x) =0$.
		 Based on \cite[Proposition 2]{chernozhukov2010quantile}, it follows under the assumptions (A1), (A2), and (A3) that the level median as a mapping from $\DD$ to $\R$ is Hadamard differentiable at $(F_\mu - F_\nu)$ for perturbations $\Delta \in \CC$ where the 
		 derivative is given by $$\DH{(F_\mu - F_\nu)} \LevMed \colon \CC \rightarrow \R, \quad \Delta \mapsto \frac{\sum_{i = 1}^{N}\Delta(t_i)/(F_\mu - F_\nu)'(t_i) }{\sum_{i = 1}^{N}1/(F_\mu - F_\nu)'(t_i) }.$$
		 By Donsker's theorem it follows that $\sqrt{n}(F_{\hat \mu_n}  - F_\mu)\konvW \Bbb_{F_\mu}$ in $\DD$ \cite[Theorem 14.3]{billingsley1999convergence} where by continuity of $F_\mu$ the Brownian bridge $\Bbb_{F_\mu}$ has a version such that almost all sample paths are in $\CC$.   
		  Applying the functional delta method  \cite[Theorem 3.9.5]{van1996weak} yields for $n\rightarrow \infty$ that \begin{multline*} \sqrt{n}\left[\Big(F_{\hat \mu_n} - F_{\nu} - \LevMed(F_{\hat \mu_n}  - F_\mu)\Big) - \Big(F_{\mu} - F_{\nu} - \LevMed(F_{\mu}  - F_\mu)\Big)\right]\konvW \\
	\Bbb_{F_\mu} -  \frac{\sum_{i = 1}^{N}\Bbb_{F_\mu}(t_i)/|(F_\mu - F_\nu)'(t_i) |}{\sum_{i = 1}^{N}1/|(F_\mu - F_\nu)'(t_i)| } \quad \text{ in } \DD.
	\end{multline*}
	Moreover, the Hadamard derivative of the absolute value $| \cdot | \colon \DD \rightarrow \DD, f \mapsto |f| = (|f(t)|)_{t \in \Interval}$ at $G_{\mu,\nu} \coloneqq  \big(F_\mu - F_\nu - \LevMed(F_\mu - F_\nu)\big)$ is given by $$\begin{aligned}
		&\DH{(G_{\mu,\nu})} | \cdot |\colon \DD \rightarrow \DD, \quad \Delta \mapsto  \DH{(G_{\mu,\nu})} | \cdot | (\Delta)\;,\\		
		&\Big(\DH{(G_{\mu,\nu})} | \cdot | (\Delta)\Big)(t) = \begin{cases}
		|\Delta(t)| & \text{ if } t \in \{ t_1, \dots, t_N\},\\
		\sign\big(G_{\mu,\nu}(t)\big) \Delta(t) & \text{ else}.
	\end{cases} 
	\end{aligned} $$	
	Hence, by functional delta method for $| \cdot |$ and the continuous mapping theorem for the operator $\int_0^1 \cdot\; dt$ it follows for $n \rightarrow \infty $ that \begin{multline*}
		 \sqrt{n}\,\Big(\COT(\hat\mu_n, \nu) - \COT(\mu, \nu)\Big)	 \konvW \\
	 \int_{0}^{1} \left(\DH{(G_{\mu,\nu})}|\cdot |\left( \Bbb_{F_\mu}(\cdot) -   \frac{\sum_{i = 1}^{N}\Bbb_{F_\mu}(t_i)/|(F_\mu - F_\nu)'(t_i)| }{\sum_{i = 1}^{N}1/|(F_\mu - F_\nu)'(t_i)| }\right)\right)(t)dt.
	\end{multline*}
	By assumption (A3) the zeros of $G_{\mu, \nu}$ are exactly given by $\{t_1, \dots, t_N\}$ which is a null set for Lebesgue measure. Further, by definition of the level median it follows that 
	$\int_{0}^{1}\sign(G_{\mu, \nu}(t))dt =0$. This yields that the limit law is given by the centered Gaussian $\mathcal{N}(0,\sigma_{\mu|\nu}^2)$ as stated in the theorem. Finally, the sum-representation of $\sigma^2_{\mu|\nu}$ follows by a straight-forward computation. 
\end{proof}
	
\begin{remark} \label{rem:TwoSamples}Our results easily extend to scenarios $\hat\mu_n = \sum_{i = 1}^{n} \delta_ {X_i}$ and $\hat \nu_m = \sum_{j = 1}^{m} \delta_{ Y_{j}}$ based on i.i.d. samples $X_1, \dots, X_n\sim \mu$ and independently sampled $Y_1,\dots, Y_m \sim \nu$. It then holds for $n, m \rightarrow \infty$ with $m/(n+m) \rightarrow \delta \in (0,1)$ that  \begin{equation}
  \sqrt{\frac{nm}{n+m}}\Big(\big(F_{\hat \mu_n}- F_{\hat \nu_m}\big) - \big(F_\mu - F_\nu\big) \Big) \konvW \sqrt{\delta} \Bbb_{F_\mu} - \sqrt{1-\delta} \Bbb_{F_\nu} \quad \text{ in } \DD.\label{eq:WeakLimitTwoSampleCase}
\end{equation}
Notably, for $\mu = \nu$ the limit law is given by  $\sqrt{\delta} \Bbb_{F_\mu} - \sqrt{1-\delta} \Bbb_{F_\mu} \stackrel{\mathcal{D}}{=} \Bbb_{F_\mu} $. Hence, the limit law of the empirical COT distance for the two-sample case follows as an application of the functional delta method in conjunction of weak convergence as in \eqref{eq:WeakLimitTwoSampleCase}. More precisely, it holds for $\mu= \nu$ under no additional assumptions for $n,m \rightarrow \infty$ with $m/(n+m) \rightarrow \delta\in (0,1)$ that \begin{equation*}
	\sqrt{\frac{nm}{n +m}}\,\COT(\hat\mu_n, \hat\nu_m) \konvW 
	\int_{0}^{1} \left| \Bbb_{F_\mu}(t)- \LevMed(\Bbb_{F_\mu} )\right|dt.
\end{equation*}  For $\mu \neq \nu$ it follows under assumptions (A1), (A2), (A3) that \begin{equation*}
			\sqrt{\frac{nm}{n +m}}\,\Big(\COT(\hat\mu_n, \hat\nu_m) - \COT(\mu, \nu)\Big)\konvW \mathcal{N}\left(0, \sigma^2_{\delta, \mu, \nu}\right),
		\end{equation*}
		where the variance is given by $\sigma^2_{\delta, \mu, \nu} = \sqrt{\delta}\sigma_{\mu|\nu}^2 +\sqrt{1-\delta}\sigma_{\nu|\mu}^2$.
\end{remark}

\subsection{Limit Laws for Bootstrapped Circular Optimal Transport Distances}
\label{sec:LimitLawsBootstrapOTD}

Given a statistic $T(X_1, \dots, X_n)$ based on finitely many random variables $X_1, \dots,$ $X_n$, its distributional pattern is often difficult to compute exactly. Therefore, approximation methods are required.
 A simple and powerful procedure for this endeavor is to perform a bootstrap. In fact, whenever the statistic $T$ is Hadamard differentiable in a suitable sense, it follows that the naive $n$-out-of-$n$ bootstrap is consistent \cite[Theorem 3.9.11]{van1996weak}. However, for functionals that are only \emph{directionally} Hadamard differentiable \cite{Roemisch04}, i.e. when the derivative is non-linear, D\"umbgen \cite{dumbgen1993nondifferentiable} shows that this resampling technique generally fails to be consistent. Nevertheless, for this setting the  $m$-out-of-$n$ bootstrap for $m = o(n)$ remains consistent \cite[Proposition 2]{dumbgen1993nondifferentiable}. To formalize these results on bootstrap consistency we follow  \cite{van1996weak}.

Recalling the definition of empirical measures $\hat \mu_n= \frac{1}{n} \sum_{i =1}^{n} \delta_{X_i}$ based on an i.i.d. sample $X_1, \dots, X_n \sim \mu$, we introduce the empirical bootstrap measure  $\hat \mu_{n,n}^* = \frac{1}{n} \sum_{j =1}^{n} \delta_{X_j^*}$  based on an i.i.d. sample $X_1^*, \dots, X_n^* \sim \hat \mu_n$. Further, let $F_{\hat \mu_{n,n}^*}$ be the empirical bootstrap cumulative distribution function. Then it follows that  the bootstrap empirical process $\sqrt{n}(F_{\hat \mu_{n,n}^*} - F_{\hat \mu_{n}})$ conditioned on $X_1, \dots, X_n$  converges weakly towards the empirical process $\sqrt{n}(F_{\hat \mu_{n}} - F_{\mu})$ as $n$ tends to infinity \cite[Theorem 3.6.1]{van1996weak}. To make this statement precise we define \begin{multline*}
	\BL{\DD} \coloneqq \{ \Phi \colon \DD \rightarrow \R \colon \\
	|\Phi(f)| \leq 1, \,|\Phi(f) - \Phi(g)| \leq \norm{f-g}_\infty \text{ for all } f,g \in \DD\}
\end{multline*}
as the space of functionals  on $\DD$ bounded by one and  Lipschitz with modulus one. Likewise, we define the space $\BL{\R}$ of bounded Lipschitz functions on $\R$. With this notation, consistency of the $n$-out-of-$n$ bootstrap  means that the quantity
$$\sup_{\Phi \in \BL{\DD}}\left|\EE{\Phi\left( \sqrt{n}(F_{\hat \mu_{n,n}^*} - F_{\hat \mu_{n}})  \right)\Big| X_1, \dots, X_n} 
	- \EE{\Phi\left( \sqrt{n}(F_{\hat \mu_{n}} - F_{\mu})\right)}\right| $$
converges  in outer probability (with respect to $X_1, \dots, X_n$) towards zero as $n \rightarrow \infty$.  Our findings for the consistency on COT distances are summarized in the following theorem. The two-sample case can be dealt with analogously.

\begin{theorem}\label{them:Bootstrap_Consistency_EqualDistr}	
	For any probability measure $\mu$ on $\SSS$ it follows for $n,m \rightarrow \infty$ with $m=o(n)$ that 
	\begin{multline*}
	\sup_{\Phi \in \BL{\R}}\Big|\EE{\Phi\big( \sqrt{m} \COT(\hat \mu_{n,m}^*,\hat \mu_{n}  )\big)\big| X_1, \dots, X_n}	- \EE{\Phi\big( \sqrt{n}\COT(\hat \mu_{n},\mu)\big)}\Big| \xrightarrow{\Prb} 0.
\end{multline*} Furthermore, for probability measures $\mu, \nu$ on $\SSS$ that fulfill assumptions (A1), (A2), and (A3) it follows for $n \rightarrow \infty$  that 
\begin{multline*}
	\sup_{\Phi \in \BL{\R}}\Big|\EE{\Phi\big( \sqrt{n} \big(\COT(\hat \mu_{n,n}^*,\nu)  - \COT(\hat \mu_{n},\nu)\big)\big)\big| X_1, \dots, X_n}	\\[-0.1cm]
	- \EE{\Phi\big( \sqrt{n}\big(\COT(\hat \mu_{n},\nu) - \COT(\mu,\nu)\big)\big)}\Big| \xrightarrow{\Prb} 0.
\end{multline*}
\end{theorem}

\section{Simulations}\label{sec:SimulationForLimitLaws}

In order to assess the finite sample performance of our asymptotic results, we perform  Monte Carlo simulations. More precisely, we take samples of different sizes from a uniform distribution and compare the law of the COT distances between empirical measure and population counterpart with the theoretical limit distribution. Additionally, we illustrate the consistency of the $m$-out-of-$n$ bootstrap for $m= \lceil n^{0.8} \rceil$ which satisfies $m = o(n)$. 

	\begin{figure}[h!]
\centering
\includegraphics[width = 0.9 \textwidth]{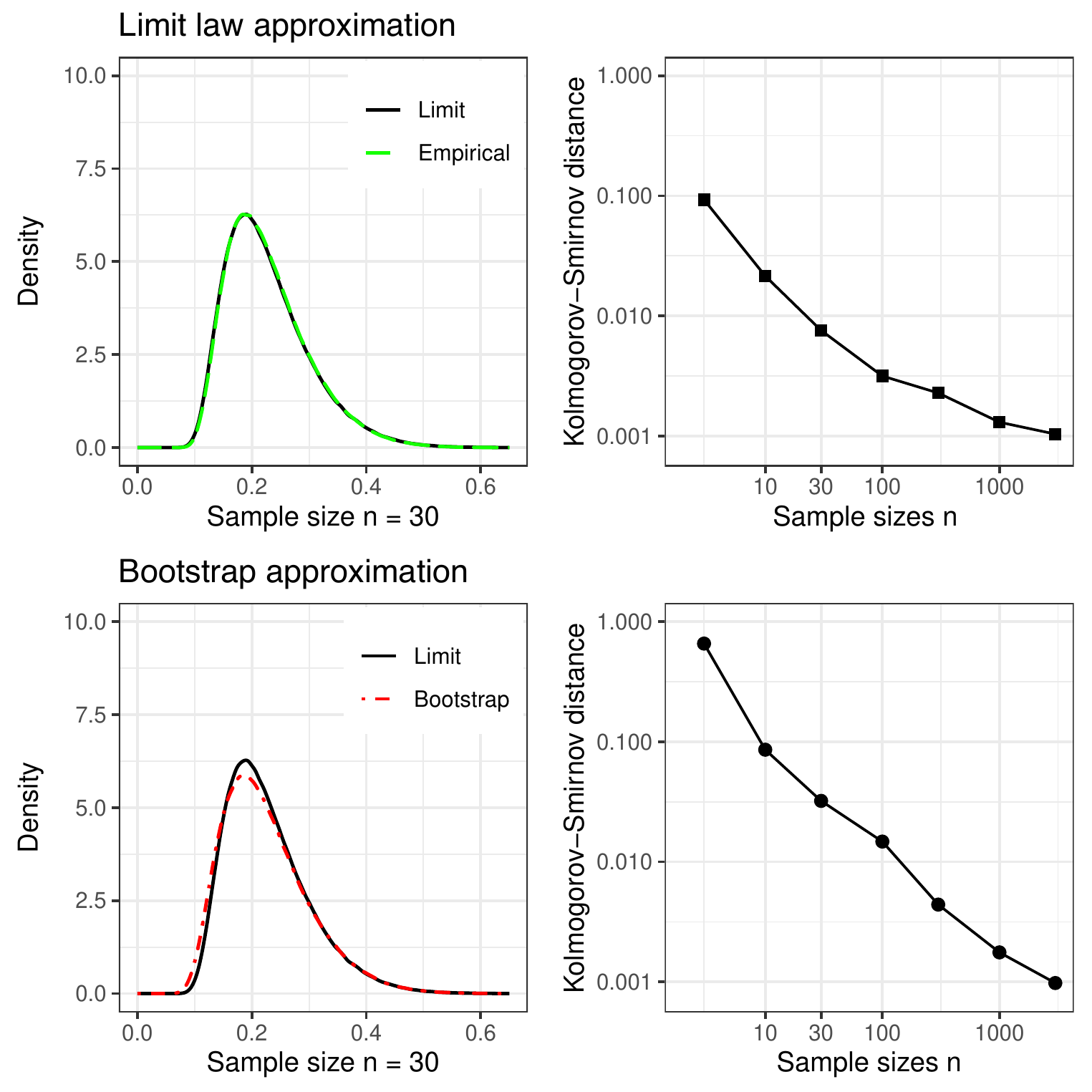}
\caption{\textbf{Accuracy of limit law for empirical and bootstrapped COT distance for $\mu = \nu = \text{Unif}(\SSS)$.} \textbf{Top:} 
Density of finite sample distribution (dashed line, green) for $\sqrt{n}\COT(\hat \mu_n, \mu)$ approximated by $10^6$ realizations each of size $n = 30$ and density of limit distribution (solid line, black) (left). The densities are approximated with a gaussian kernel and Silverman's rule \cite{silverman1986density}.
The Kolmogorov-Smirnov distance between empirical distribution and limit law for different sample sizes $n \in \{3,10,30,100,300,1000,3000\}$ on a logarithmic scale (right).
\textbf{Bottom:} Same setting as top, where instead bootstrapped COT distances (dot-dashed line, red) from an $m$-out-of-$n$ bootstrap with $m =  \lceil n^{0.8} \rceil$, i.e. for $n = 30$ and $m = 16$ on the left, are compared to the limit law (solid line, black).}
\label{fig:LimitLawApprox}    
\end{figure}

The simulations are carried out with the software \textsc{R} \cite{Rcore2020} and are depicted in Figure \ref{fig:LimitLawApprox}. For computation of COT distances on $\SSS$, we employ the discretization scheme from Section \ref{sec:CircOT} for $D= 1000$.
To generate samples from the limit law, we discretize the brownian Bridge $\Bbb_{F_\mu}$ at the locations $\{i/D \colon i \in \{1, \dots, D\}\}$, i.e.  $\Bbb_{F_{\tilde\mu_D}}(t) \coloneqq \Bbb_{F_\mu}(\lceil D t \rceil /D )$ and use the approximation $$ \int_0^1 \left| \Bbb_{F_\mu}(t) - \LevMed(\Bbb_{F_\mu}) \right|dt \\
	\approx \frac{1}{D} \sum_{i = 1}^{D}\left| \Bbb_{F_\mu}(i/D) - \LevMed\big(\Bbb_{F_{\tilde\mu_D}}\big) \right|.$$  

Our simulations in Figure \ref{fig:LimitLawApprox} show that the law of the empirical COT distance $\sqrt{n}\COT(\hat \mu_n, \mu)$ matches its limit distribution fairly well even for small sample sizes $(n = 30)$. Furthermore, the $m$-out-of-$n$ bootstrap also appears to be consistent for COT distances which is in line with our theoretical results. 
For the setting $\mu\neq \nu$ (e.g. two different von Mises distributions), we observe in our simulations a similar performance of approximating the corresponding Gaussian distribution, hence the details are omitted here.

\section{Testing for Goodness of Fit}\label{sec:TestingGoF}

Many popular statistical tests such as goodness of fit tests are based on the notion of a distance between probability measures (see e.g. Kolmogorov-Smirnov, Cramer-von Mises, Maximum mean discrepancy). Their aim is to investigate  
 whether a given sample is taken from a particular probability measure $\mu_0$. To formalize this concept, let $X_1, \dots, X_n \sim \mu$ be an i.i.d. sample. Based on the data, our aim is to test the hypothesis
$$\mathcal{H}_0\colon \mu = \mu_0.$$
 Herein, we propose the following COT based test. 

\begin{Test}[COTT]\label{test:OT-Test}
Let $\alpha\in (0,1]$ and denote $\hat\mu_n$ as the empirical measure for the sample $X_1, \dots, X_n$. We reject $\mathcal{H}_0$ with significance level $\alpha$ if $$\sqrt{n}\COT(\hat\mu_n, \mu_0) > q_{1-\alpha},$$
 where $q_{1-\alpha}$ is the $(1-\alpha)$-quantile of the distribution of the random variable $\int_{0}^1 | \Bbb_{F_{\mu_0}}(t) - \LevMed(\Bbb_{F_{\mu_0}})| dt$. 
\end{Test}
Our proposed test exhibits a natural interpretation which is based  on the OT plan for the COT problem. Intuitively, the more difficult it is to transport all probability mass from the empirical measure $\hat\mu_n$ onto the null distribution $\mu_0$, the less likely it is that the associated sample is drawn from $\mu_0$. Consequently, if $\COT(\mu_0, \mu)$ is large, we expect that our proposed test rejects with high probability. Later in Figure \ref{fig:OT_Distances} and the surrounding text a more detailed explanation is provided.

\begin{theorem}[Consistency of COTT] For any $\alpha>0$ and probability measures $\mu_0 \neq \mu_1$ on $\SSS$ it holds as $n$ tends to infinity that $$\begin{aligned}
		\Prb_{\mu_0}\left( \sqrt{n}\COT(\hat\mu_n, \mu_0) > q_{1-\alpha}\right) \rightarrow \alpha \quad \text{ and } \quad 
		\Prb_{\mu_1}\left( \sqrt{n}\COT(\hat\mu_n, \mu_0) > q_{1-\alpha}\right) \rightarrow 1.
	\end{aligned}$$
\end{theorem}

\begin{proof}
	The first assertion follows from Theorem \ref{them:LimitLawEqualDistr}. For the second assertion, we note under $\mu = \mu_1$ that $F_{\hat \mu_n} \rightarrow F_{\mu_1}$ in $\DD$ almost surely, as $n\rightarrow \infty$ \cite{van1996weak}. By Lipschitz property of the level median with respect to supremum norm \cite[Theorem~1]{avenski2019LimitMonotoneArrangement}, we see 
	for $n \rightarrow \infty$ by the continuous mapping theorem \cite{vaart_1998} that $\COT(\hat \mu_n, \mu_0) \rightarrow \COT(\mu_1, \mu_0)>0$ almost surely.
	Hence, it follows  that $\sqrt{n}\COT(\hat \mu_n, \mu_0) \rightarrow\infty$ almost surely, which implies the second claim.
\end{proof}
For the uniform distribution on $\SSS$ as well as certain von Mises distributions we include in Table \ref{tab:QuantilesVM} the associated $(1-\alpha)$-quantiles $q_{1-\alpha}$ for $\alpha \in \{0.1, 0.05, 0.01\}$. All critical values are obtained via Monte Carlo simulations using our implementation of the COTT in our $\textsc{R}$-package \texttt{circularOT}. For the cumulative distribution function of the von Mises distributions as well as random number generation we use the $\textsc{R}$-package \texttt{circular} \cite{RcircularPackage2017}. For other null distributions $\mu_0$ the quantile $q_{1-\alpha}$ may be approximated through similar Monte Carlo simulations as described in Section \ref{sec:SimulationForLimitLaws}. Alternatively, given a sample of size $n$ from $\mu_0$ an $m$-out-of-$n$ bootstrap for $m = o(n)$ may be applied to estimate the quantile $q_{1-\alpha}$.

\begin{table}[!t]
\centering
\caption{Critical values $q_{1-\alpha}$ of COT test for  $\alpha \in \{0.1, 0.05, 0.01\}$ obtained through Monte Carlo simulations for $\mu_0$  a von Mises distribution (Example \ref{eq:VonMisesAssumptionA3}) with concentration parameter $\kappa \in \{0, 0.5, 1,2,3\}$.
For $\kappa=0$ the von Mises distribution is equal to the uniform distribution on $\SSS$\!. For each $\kappa$ in total $N = 10^6$ realizations are drawn from the theoretical limit distribution from Theorem \ref{them:LimitLawEqualDistr} with a discretization scheme for $D= 1000$ as described in Section \ref{sec:SimulationForLimitLaws}.\newline}
\label{tab:QuantilesVM}

\begin{tabular}{p{1.22cm}p{1.2cm}p{1.2cm}p{1.2cm}p{1.2cm}p{0.8cm}}
\hline\noalign{\smallskip}
$\kappa$     & 0     & 0.5   & 1     & 2     & 3     \\ 
\hline\noalign{\smallskip}
$q_{0.9}$ \quad  &0.327 \quad  & 0.318 \quad & 0.295 \quad & 0.238\quad  & 0.194  \\
$q_{0.95}$ \quad &0.367 & 0.357 & 0.330  & 0.267 & 0.219  \\
$q_{0.99}$ \quad &0.447 & 0.434 & 0.403 & 0.328 & 0.271 \\
\noalign{\smallskip}\hline\noalign{\smallskip}
\end{tabular}
\end{table}

\subsection{Testing for Uniformity}
As an illustrative example, we employ 
the COTT in order to test for uniformity \cite{Garcia2018OverviewTestsHyptersphere, landler2018circular, landler2019hermans}. The respective hypothesis is $$\label{eq:HypothesisForUniformity}
  \mathcal{H}_0 \colon \mu = \text{Unif}(\SSS).$$
To investigate the performance of COTT for testing of uniformity, we compare it with other prominent proposals. Notably, some of those are specifically tailored to perform well for unimodal alternatives but lack statistical power in case of multimodal alternatives  \cite{bergin1991comparison}. To incorporate this aspect in our analysis, we first test for uniformity against von Mises distributions. In this setting, Rayleigh's test is known to be the most powerful test  \cite{watson1956construction} and therefore serves as a benchmark. Afterwards, we test against Stephens' multimodal distributions \cite{Stephens1969MultimodalDistr} which we introduce in~\eqref{eq:StephensAlternative1}.

\subsection{Power Analysis under von Mises Alternatives}
We assess the performance of COTT in case of unimodal alternatives by considering different von Mises distributions (see Example~\ref{expl:vonMisesDistributions}) with mean $\theta = 0.5$ and varying concentration parameter $\kappa\in \{0, 0.1, 0.2, \dots, 2.5\}$. Figure \ref{fig:Densities} (top plot) illustrates such densities for certain $\kappa$. 
 We generate 10,000 repetitions each of sample size $n = 30$ and compute the empirical power, i.e. the rejection probability of COTT on uniformity for significance level $\alpha = 0.05$. For comparison to other well-known tests, we also determine the empirical power of Rayleigh's test \cite{rayleigh1880},  Kuiper's test \cite{Kuiper1960TestsCR}, Watson's test \cite{watson1961goodness}, Rao's range and spacing tests \cite{rao1969some}, as well as some more recently proposed tests by Pycke \cite{pycke2010some}. 
 In accordance with the notation by Pycke \cite{pycke2010some}, we consider his  proposed tests based on the test statistics $V_{0.1}, V_{\sqrt{1/2}}, V_{\sqrt{2/3}}, V_{\sqrt{3/4}}$, and~$G$.
 
The empirical rejection probabilities of all these tests for various concentration parameters $\kappa$ are computed with the software \textsc{R} \cite{Rcore2020}. For random number generation of von Mises distributions and implementations of tests by Rayleigh, Kuiper, Watson, and Rao we use the package \texttt{circular}. Concerning the COTT we employ the implementation from our package \texttt{circularOT}. 
Results are depicted in Figure \ref{fig:TestingAgainstUnimodal} (top plot). In summary, all tests keep the level for $\kappa = 0$. As the  concentration parameter $\kappa >0$ increases, the rejection probability of each test also increases.

Rayleigh's test, the most powerful test for this setting, and Pycke's $V_{0.1}$-test perform best. These findings are in line with empirical observations by Pycke \cite{pycke2010some} as the $V_{0.1}$-test is specifically designed against unimodal alternatives. Watson's test and the COTT perform almost as well and essentially exhibit the same empirical power for different values of $\kappa$ when compared to each other. 
Let us note that the COTT can be understood as an $L^1$-version of Watson's test where the test statistic for a given sample with empirical measure $\hat \mu_n$ is given by \begin{align*}	U_n^2 &= \inf_{\alpha \in \R}\int_{0}^1\left( F_{\hat \mu_n}(x)  - x - \alpha\right)^2dx
	\\&=\int_{0}^1\left( F_{\hat \mu_n}(x)  - x - \int_{0}^1 (F_{\hat \mu_n}(y)  - y)dy  \right)^2dx.
\end{align*}
This may explain their similar performance.  All remaining tests exhibit lower rejection probabilities, in particular Rao's tests display the smallest statistical power.
\begin{figure}[t!]
\centering
  \includegraphics[width=0.9\textwidth]{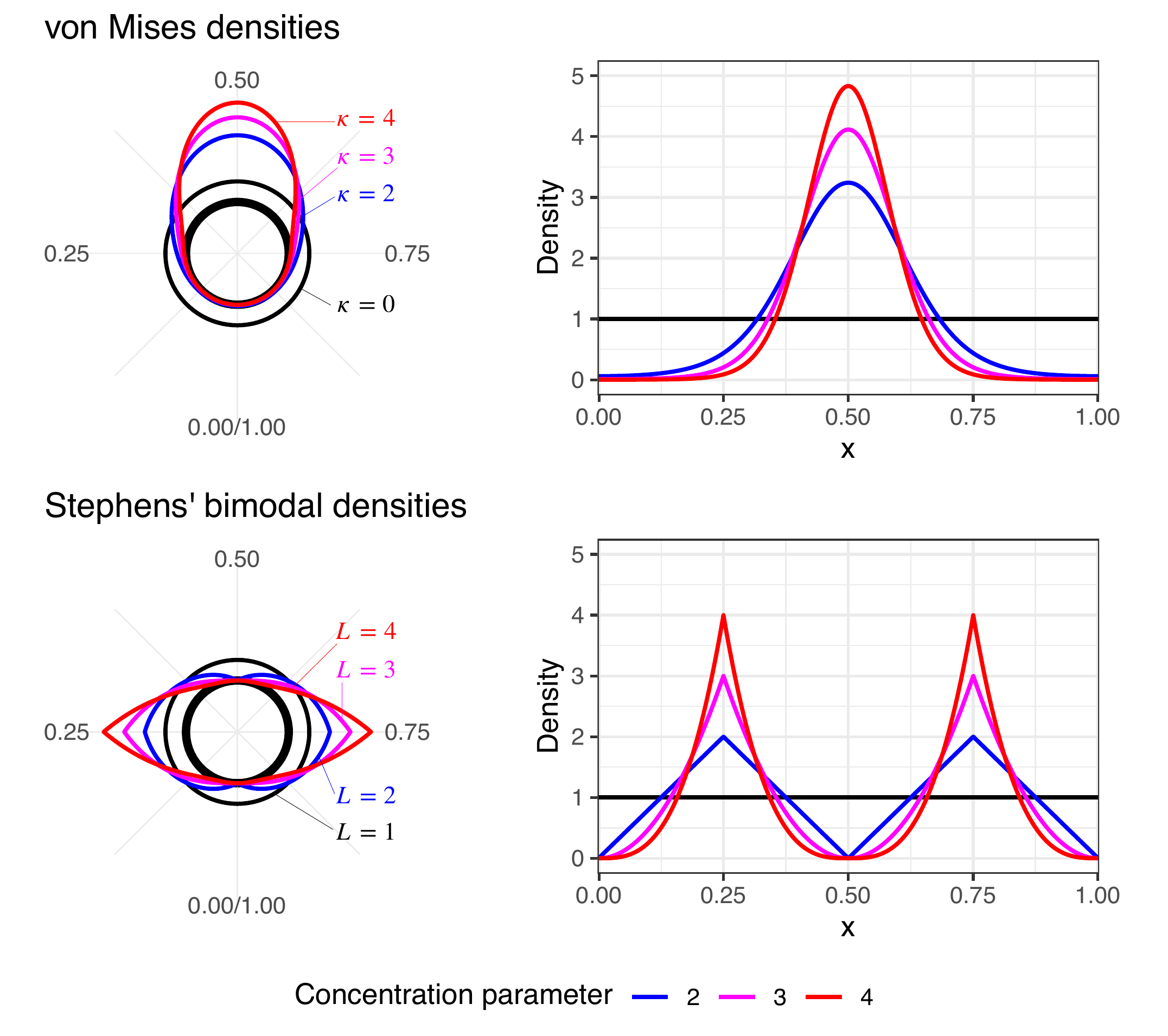}
  \caption{\textbf{Densities of von Mises and Stephens' distributions.}   \textbf{Top:}  Densities for von Mises distributions (Example \ref{expl:vonMisesDistributions}) with mean $\theta = 0.5$ for different concentration parameters $\kappa \in \{2,3,4\}$ and the density of the uniform law (black) in a circular plot (left) and a cartesian plot (right). 
  \textbf{Bottom:} Same setting as top for Stephens' bimodal distributions \eqref{eq:StephensAlternative1} with concentration parameter $L \in \{2,3,4\}$. }
   \label{fig:Densities}
\end{figure}

\begin{figure}[t!]
\centering
\includegraphics[width=0.9 \textwidth, trim = 10 0 0 0, clip]{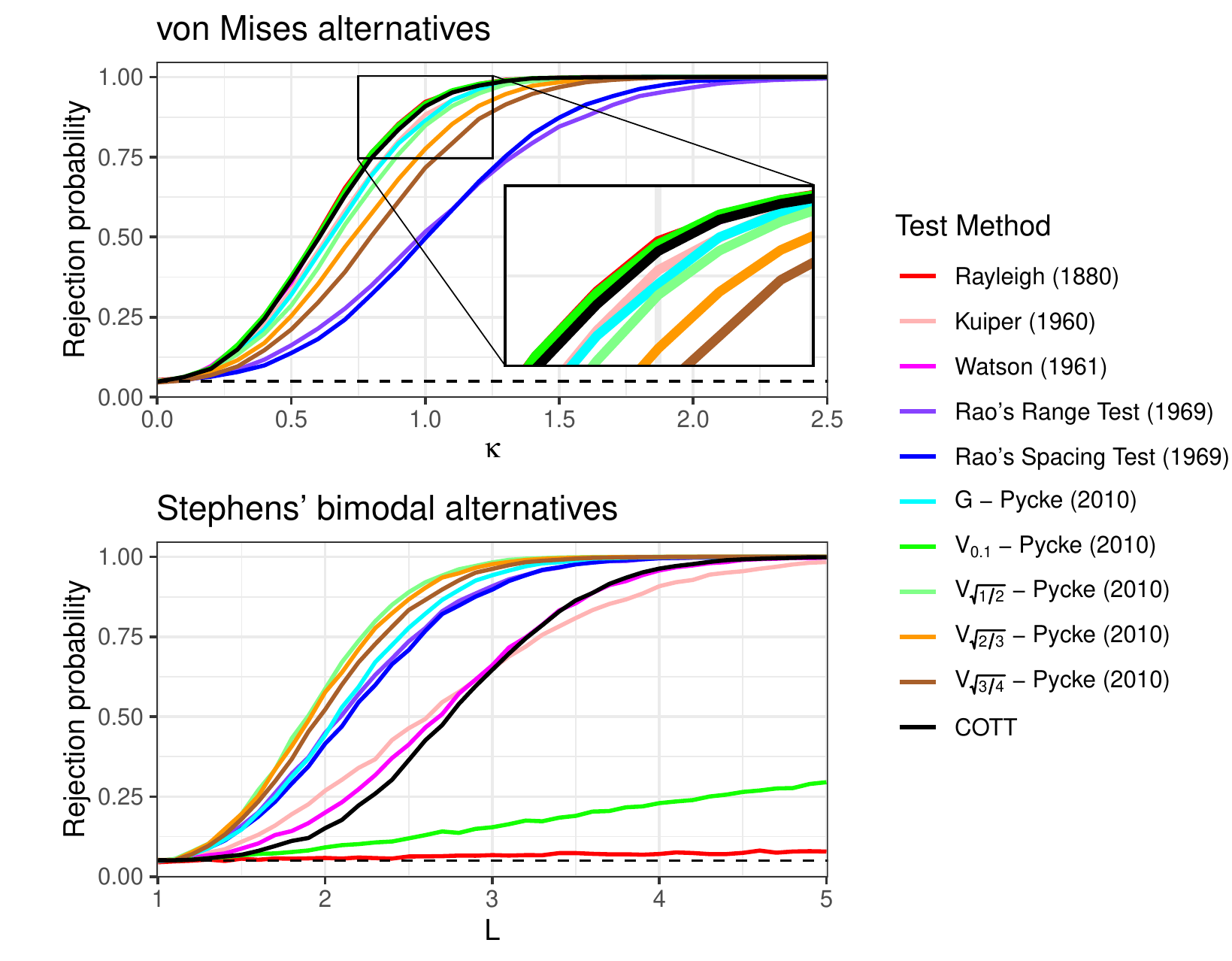}
\caption{\textbf{Statistical power of tests for uniformity under von Mises and Stephens' bimodal alternatives.} \textbf{Top:} Empirical rejection probabilities for tests on uniformity with significance level $\alpha = 0.05$ based on 10,000 repetitions of sample size $n = 30$ from von Mises distributions with mean $\gamma = 0$ and concentration parameter $\kappa\in \{0,0.1, \dots, 2.5\}$. The dashed black line represents the level $\alpha = 0.05$. 
\textbf{Bottom:} Empirical rejection probabilities with significance level $\alpha = 0.05$ where instead 10,000 repetitions of sample size $n = 30$ are taken from Stephens' bimodal distributions, i.e. $M = 2$ with concentration parameter $L\in \{1, 1.1, \dots, 5\}$.
}
\label{fig:TestingAgainstUnimodal}  
\end{figure}

\subsection{Power Analysis under Stephens' Multimodal Alternatives}

For a comparison of the different tests for uniformity in case of multimodal alternatives, let us introduce Stephens' multimodal distribution \cite{Stephens1969MultimodalDistr} characterized for $M\in \N$ and $L\geq 0$ by the density
\begin{equation}\label{eq:StephensAlternative1}
P^{S}_{M,L}(x) \coloneqq \begin{cases}
L(2Mx)^{L-1} & \text{ if } 0\leq x < \frac{1}{2M},\\
L(2M - 2Mx)^{L-1} & \text{ if } \frac{1}{2M}\leq x < \frac{1}{M},\\
P^{S}_{M,L}( x - \lfloor xM \rfloor /M )&\text{ if } \frac{1}{M} \leq x < 1.
\end{cases}
\end{equation}
The densities for the bimodal case, i.e. $M = 2$, and certain values for $L$ are shown in Figure \ref{fig:Densities} (bottom plot). 
Note that for  $L = 1$, Stephens distributions coincide for all values of $M \in \N$ with the uniform distribution on $\SSS$. In case $L>1$, the parameter $M$ describes the number of modes whereas the parameter $L$ indicates the concentration of mass towards these modes (spikiness). 

For our analysis of power for all previously stated tests, we consider 10,000 repetitions each of sample size $n = 30$ from Stephens' bimodal distributions, i.e. for fixed $M = 2$ with varying $L \in \{1, 1.1, \dots, 5\}$ and compute the respective empirical rejection probability. Each sample is tested for uniformity with significance level $\alpha = 0.05$ using the described methods. 

\begin{figure}[t!]
\centering
\includegraphics[width = 0.9 \textwidth, trim = 10 0 0 0, clip ]{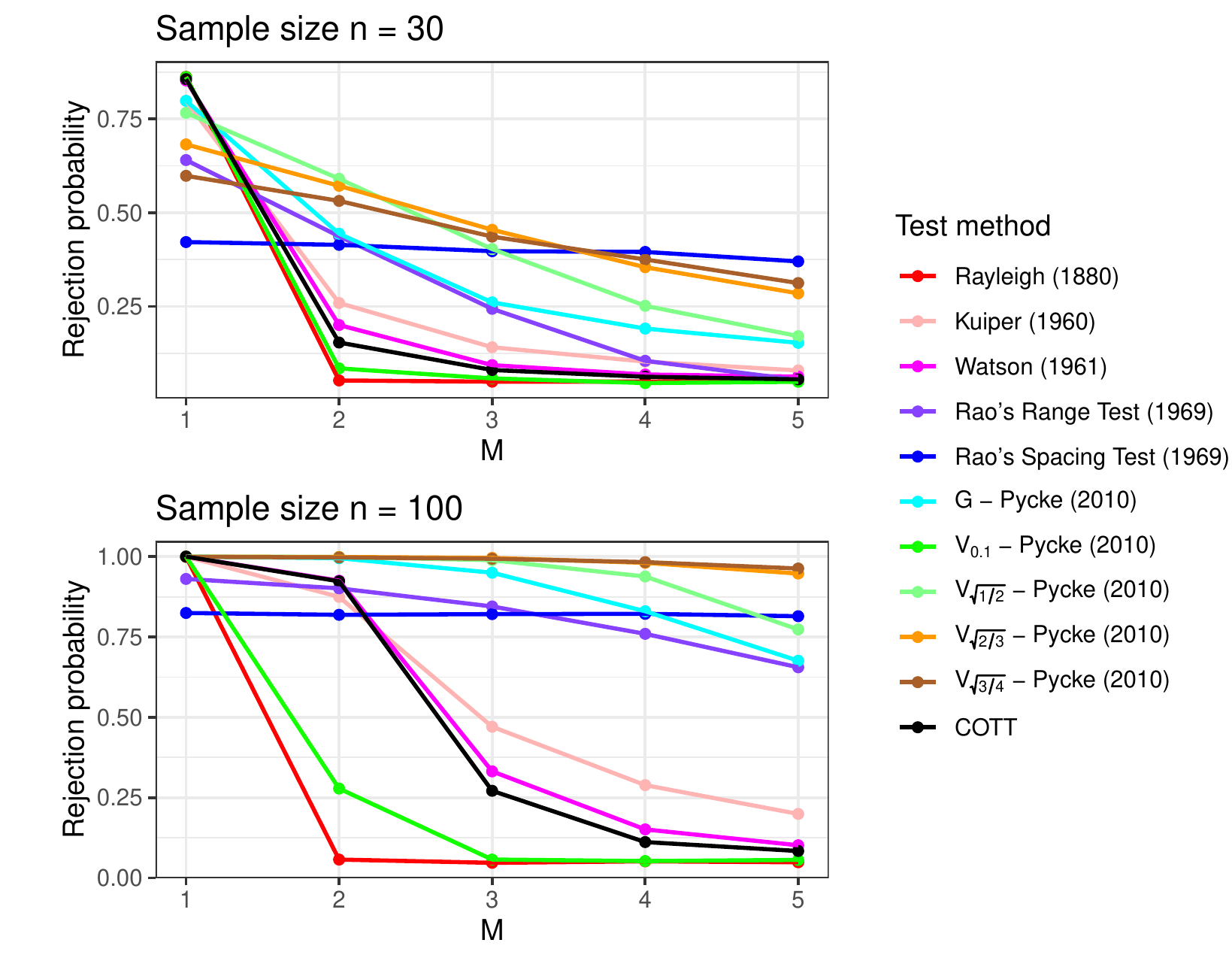}
\caption{\textbf{Statistical power of tests for uniformity under Stephens' multimodal distributions with $L =2$ and different number of modes.} Empirical rejection probabilities for tests on uniformity based on 10,000 repetitions of sample size $n = 30$ (top) and $n = 100$ (bottom) with significance level $\alpha =0.05$.
}
\label{fig:TestingAgainstMultipleModes}  
\end{figure}

The  rejection probabilities for these test scenarios are illustrated in  Figure \ref{fig:TestingAgainstUnimodal} (bottom plot). Under the null hypothesis all tests keep the level, and for increasing concentration parameter $L$ the rejection probability of each test increases. Overall, Pycke's $V_{\sqrt{1/2}}$-test performs best. Pycke's $V_{\sqrt{2/3}}$-, $V_{\sqrt{3/4}}$-, and $G$-test as well as both tests by Rao perform almost as well. In contrast, Pycke's $V_{0.1}$-test and Rayleigh's test, which are both known to perform well for unimodal alternatives,  feature by far the smallest statistical power.
 The COTT as well as Kuiper's and Watson's test all exhibit a fairly similar power and reject slightly less often than Rao's tests.
\begin{figure}[t!]
\centering
  \includegraphics[width=0.9\textwidth]{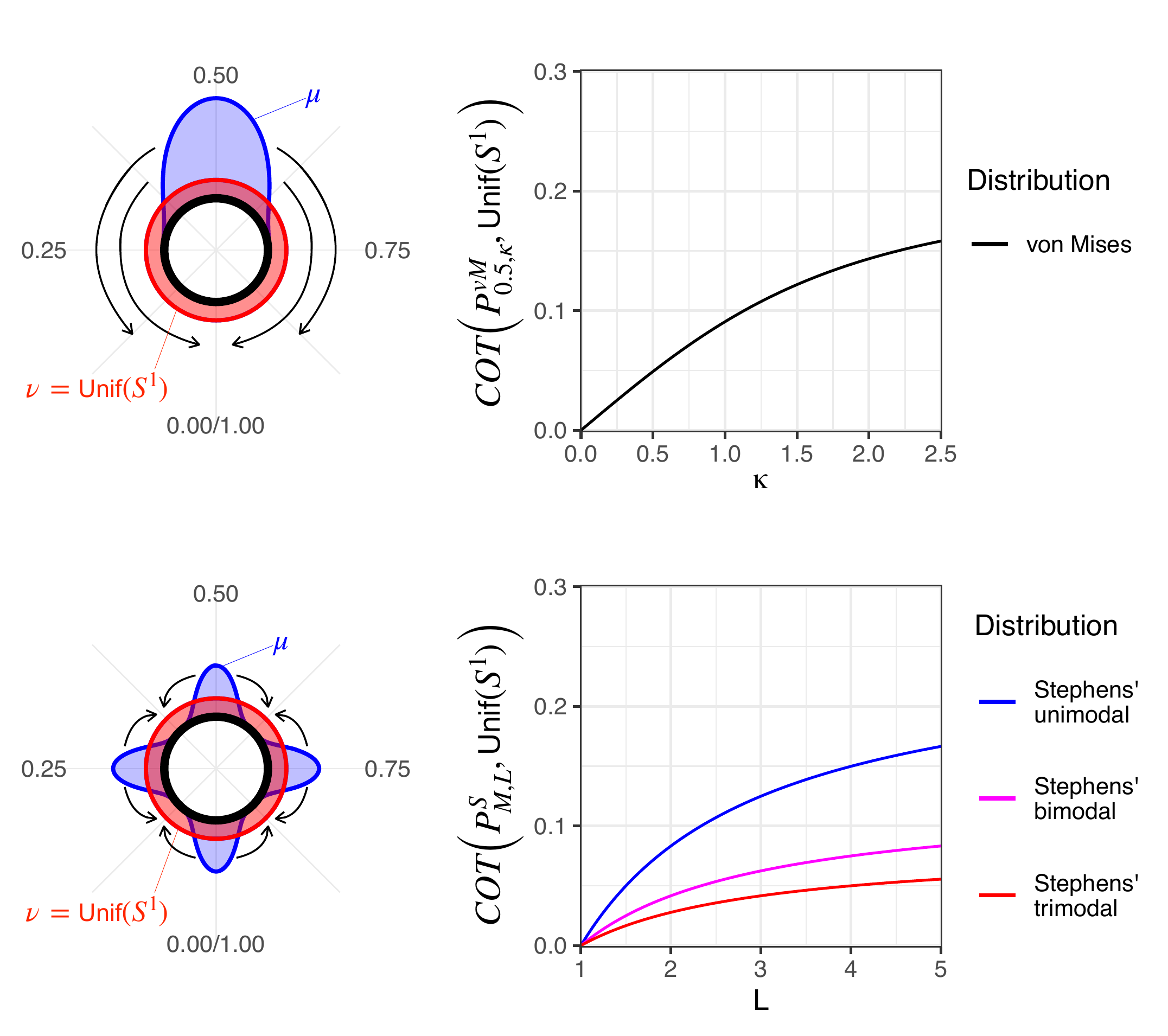}
  \caption{\textbf{Optimal mass transportation from single and multiple mode distributions to a uniform distribution, respectively.}
  \textbf{Top:} A large portion of the probability mass from $\mu$ has to be transported for a relatively long distance to match with the uniform distribution $\nu$ (left).
  COT distance between uniform and  von Mises distribution (black) with mean $\theta = 0.5$ and $\kappa \in [0,5]$ (right). Notably, for $\kappa=0$ the resulting von Mises distribution coincides with the uniform on $\SSS$. 
    \textbf{Bottom:}  Probability mass of $\mu$ spreads only locally, transport distances are overall shorter (left). 
	COT distance between uniform and Stephens' distributions for unimodal (blue), bimodal (pink), and trimodal (red) setting and concentration (right). For $L=1$ the respective Stephens distribution is equal to the uniform distribution on $\SSS$.}   
  \label{fig:OT_Distances}
\end{figure}

To further investigate the effect of multimodality on the described tests,  
we consider 10,000 repetitions of sample size $n = 30$ as well as $n = 100$ from Stephens' distributions with constant concentration parameter $L = 2$ and varying number of modes $M \in \{1, \dots, 5\}$ to test for uniformity with significance level $\alpha = 0.05$. The resulting rejection probabilities of our simulations are shown in Figure \ref{fig:TestingAgainstMultipleModes}.
Notably, for most tests with the exception of Rao's spacing test the rejection probability decreases with increasing number of modes. Whereas for samples of size  $n = 30$ the tests by Rayleigh, Kuiper, Watson, Pycke's $V_{0.1}$-test as well as our COTT perform best for unimodal alternatives, they reject  much less in case of multimodal alternatives. In contrast, the remaining tests reject less often in case of unimodal alternatives but appear to be more robust when the number of modes increases.
For sample size $n = 100$, almost all tests reject with high probability in case of the unimodal alternative. Only both of Rao's tests reject less often. Concerning bimodal alternatives,  Watson's tests and our proposed test recognize non-uniformity fairly well. They even assert with a larger probability than Kuiper's test and Rao's tests that samples from Stephens' bimodal distribution do not stem from a uniform distribution.
Rayleigh's test and  Pycke's $V_{0.1}$-test again reject with a considerably lower probability for multimodal alternatives. Overall, Pycke's $V_{\sqrt{2/3}}$- and $V_{\sqrt{3/4}}$-tests perform best against multimodal alternatives, his $V_{\sqrt{1/2}}$- and $G$- test reject slightly less often.

Let us give an intuitive explanation as to why the COTT performs less effectively for bimodal, or more generally, multimodal distributions. Given a sample from a highly concentrated multimodal distribution, it is much less costly to transform the probability mass of the empirical measure $\hat \mu_n$ to the uniform distribution $\mu_0$ on $\SSS$ as the data is already sufficiently spread. Hence, the COT distance $\COT(\hat \mu_n, \mu_0)$ between uniform and Stephens' multimodal distributions is rather small (see Figure~\ref{fig:OT_Distances} (bottom plots)). Consequently, the COTT might not be able to detect this. In contrast to that, for a sample from a unimodal distribution transporting the empirical measure to a uniform measure leads to much larger transportation costs as the data is not spread but rather concentrated towards the single mode (see Figure~\ref{fig:OT_Distances} (top plots)). As a result, the associated COT distance is likely to be larger which can be observed for von Mises and Stephens' unimodal alternative in Figure~\ref{fig:OT_Distances} (right plots), thus leading to a higher rejection probability.

\section{Discussion and Outlook}\label{sec:summary}
Our simulations show that the COTT for uniformity performs particularly well for unimodal alternatives and even exhibits almost the same power as Rayleigh's tests, the most powerful test for von~Mises alternatives. Concerning multimodal alternatives the COTT is less powerful, an observation which is in line with the basic principle of OT.  Overall, if a unimodal or multimodal distribution is expected with mainly one pronounced mode  in the alternative, we recommend applying the COTT for testing of uniformity.

As a natural extension of testing a single-element null hypothesis one may seek to use the COTT to assess the goodness of fit for families of distributions, e.g. von Mises families with estimated parameters. For this purpose, the parameters have to be estimated additionally which affects the limit law and is left open for future research. Moreover, we stress that extensions to bivariate i.i.d. samples  with marginals $\mu$ and $\nu$ can be proved analogously to our results \cite{freitag2007nonparametric}.
 Beyond this aspect it seems worthwhile to examine statistical properties of COT for other ground costs, e.g. $c(x,y) = \rho_\SSS^p(x,y)$ for $p \neq 1$. Indeed, for $p>1$ a similar reduction principle of the COT problem as in \eqref{eq:OTFormula} to an optimization problem in only one unknown is available \cite{delon2010fast}. However, the characterization of the optimal value in terms of quantile or distribution function is unknown. 

Finally, extension to higher dimensional  spheres $S^d$ for $d\geq 2$ remains a challenging task. Under no further assumptions  analyzing the asymptotic fluctuation of empirical spherical OT distances seems fairly difficult due to a lack of an explicit formula for the OT distance.  
An appropriate setting might take the geodesic distance $\rho_{S^d}$ on $S^d$ as the cost function with probability measures $\mu$, $\nu$ on $S^d$ that are rotationally invariant around a common axis $\eta\in S^d$. Parametrizing the elements $p\in S^d$ by $(\theta, \sin(\theta) q)$ with $\theta= \rho_{S^d}(\eta, p)$ and $q \in S^{d-1}$ it follows that $\mu$, $\nu$ are characterized by their cumulative distribution functions with respect to $\theta$, denoted by $F_\mu^{(\eta)}, F_\nu^{(\eta)}$, respectively. 
 As in the case of the real line a closed formula for the spherical OT distance can be proven using the monotone coupling between $\mu$ and $\nu$ along the direction $\eta$ in conjunction with the dual formulation of OT (see \cite{villani2008optimal}), resulting in
  $$OT_{S^d}(\mu, \nu) = \int_{0}^{\pi}\left|F_\mu^{(\eta)}( \theta) - F_\nu^{(\eta)}( \theta)\right|d \theta. $$
	Hence, assuming the common axis is known and considering a rotationally symmetric empirical estimator along $\eta$ for $\mu=\nu$, the limit distribution of the empirical spherical OT distance scaled with $\sqrt{n}$ is equal to an integral of the absolute value of a (time-changed) Brownian bridge. However, when the common axis is unknown and has to be estimated, the analysis of the asymptotics becomes much more involved.

\section*{Acknowledgement}
The authors gratefully acknowledge support for the DFG Research Training Group
2088 \emph{Discovering Structure in Complex Data: Statistics Meets Optimization and Inverse Problems} and the DFG Cluster of
Excellence 2067 \emph{Multiscale Bioimaging: From Molecular Machines to Networks of Excitable Cells}.

\end{document}